%% file: bipartite_freq_alloc.tex
\documentclass[11pt]{article}

\usepackage{fullpage}

\usepackage{latexsym,amssymb,amsthm,amsbsy,amsmath,amsfonts,epsf,graphicx,picinpar}

\input{macros.tex}

\title{Better Bounds for Incremental Frequency Allocation\\
in Bipartite Graphs}

\author{%
       Marek Chrobak\thanks{%
       Department of Computer Science,
       University of California,
       Riverside, CA 92521, USA.
       Research supported by NSF Grant CCF-0729071.
       }
       \and
	{\L}ukasz Je{\.z}\thanks{%
		Institute of Computer Science,
		University of Wroc{\l}aw,
		ul. Joliot-Curie 15, PL-50-383 Wroc{\l}aw, Poland.
		Email: {\tt lje@cs.uni.wroc.pl}.
		Research supported by MNiSW Grant N N206 490638 2010--2011.
	}
       \and
		Ji\v{r}\'\i\ Sgall\thanks{%
Dept.\ of Applied Mathematics, Faculty of Mathematics and Physics,
Charles University, Malostransk\'e n\'am.\ 25, CZ-11800 Praha 1, Czech
Republic. 
Email: {\tt sgall@kam.mff.cuni.cz}.
Partially supported by Inst.\ for Theor. Comp. Sci.,
Prague (project 1M0545 of M\v{S}MT \v{C}R)
and grant IAA100190902 of GA AV \v{C}R.
		}
   }

\begin{document}

\maketitle

\begin{abstract}
We study frequency allocation in wireless networks.
A wireless network is modeled by an undirected graph, with vertices
corresponding to cells. In each vertex we have a certain number of
requests, and each of those requests must be assigned a different frequency.
Edges represent conflicts between cells, meaning that frequencies in
adjacent vertices must be different as well. The objective is to
minimize the total number of used frequencies.

The offline version of the problem is known to be {\NP}-hard. In
the incremental version, requests for frequencies arrive over time and 
the algorithm is required to assign a frequency to a request
as soon as it arrives. Competitive incremental algorithms have
been studied for several classes of graphs. For paths,
the optimal (asymptotic) ratio is known to be $4/3$,
while for hexagonal-cell graphs it is between $1.5$ and $1.9126$. 
For $\xi$-colorable graphs, the ratio of $(\xi+1)/2$ can be achieved. 

In this paper, we prove nearly tight bounds on the asymptotic
competitive ratio for bipartite graphs, showing that it is between $1.428$ 
and $1.433$. This improves the previous lower bound of $4/3$ and upper
bound of $1.5$.
Our proofs are based on reducing the incremental problem to a purely 
combinatorial (equivalent) problem of 
constructing set families with certain intersection properties.
\end{abstract}

%%%%%%%%%%%%%%%%%%%%%%%%%%%%%%%%%%%%%%%%%%%%%%%%%%%%%%%%%%%%%%%%%%%%%%
%%%%%%%%%%%%%%%%%%%%%%%%%%%%%%%%%%%%%%%%%%%%%%%%%%%%%%%%%%%%%%%%%%%%%%
%%%%%%%%%%%%%%%%%%%%%%%%%%%%%%%%%%%%%%%%%%%%%%%%%%%%%%%%%%%%%%%%%%%%%%

\section{Introduction}
\label{sec: introduction}

\paragraph{Static frequency allocation.}
In the frequency allocation problem, we are given a wireless network
and a collection of requests for frequencies. The network is 
modeled by a (possibly infinite) undirected graph $G$, whose vertices
correspond to the network's cells. Each request is associated with a vertex, 
and requests in the same vertex must be assigned different frequencies.
Edges represent conflicts between cells, meaning that frequencies in
adjacent vertices must be different as well. The objective is to
minimize the total number of used frequencies.
We will refer to this model as \emph{static}, as it
corresponds to the scenario where the set of
requests in each vertex does not change over time.

A more rigorous formulation of this static frequency allocation problem is
as follows: Denote by $\ell_v$ the \emph{load} at a vertex $v$ of $G$, that is
the number of frequency requests at $v$. A frequency allocation is
a function that assigns a set $L_v$ of frequencies (represented, say, by
positive integers) to each vertex $v$
and satisfies the following two conditions: 
(i) $|L_v| = \ell_v$ for each vertex $v$, and
(ii) $L_v\cap L_w = \emptyset$ for each edge $(v,w)$.
The total number of frequencies used is $|\bigcup_{v\in G} L_v|$, and this
is the quantity we wish to minimize. We will use notation
$\opt(G,\barell)$ to denote the minimum number of frequencies 
for a graph $G$ and a demand vector $\barell$.

If one request is issued per node, then $\opt(G,\barell)$ is equal to the
chromatic number of $G$, which immediately implies that the frequency
allocation problem is {\NP}-hard. In fact,
McDiarmid and Reed \cite{McDRee00} show that the problem remains {\NP}-hard
for the graph representing the network whose cells are regular
hexagons in the plane, which is a commonly studied abstraction of
wireless networks. (See, for example,
the surveys in~\cite{Murthy-etal99,Aardal-etal03}).
Polynomial-time $\fourthirds$-approximation algorithms for this
case appeared in \cite{McDRee00} and \cite{NarShe01}.

%%%%%%%%%%%%

\paragraph{Incremental frequency allocation.}
In the incremental version of frequency allocation, requests arrive over time and 
an incremental algorithm is required to assign frequencies to requests
as soon as they arrive. An incremental algorithm $\calA$ is called
\emph{asymptotically $R$-competitive} if, for any graph $G$ and load vector $\barell$,
the total number of frequencies 
used by $\calA$ is at most $R\cdot \opt(G,\barell)+ \lambda$, where $\lambda$ is a constant
independent of $\barell$. We allow $\lambda$ to depend on the class of graphs under 
consideration, in which case we say that $\calA$ is $R$-competitive for this class.
We refer to $R$ as the \emph{asymptotic competitive ratio} of $\calA$. As in this paper
we focus only on the asymptotic ratio, we will skip the word
``asymptotic" (unless ambiguity can arise), and simply use
terms ``$R$-competitive" and ``competitive ratio" instead. Following
the terminology in the literature (see \cite{ChChYZ07,ChChYZ10}, for example),
we will say that the 
competitive ratio is \emph{absolute} when the additive constant
$\lambda$ is equal $0$.

Naturally, research in this area is concerned with designing
algorithms with small competitive ratios for various classes of graphs, 
as well as proving lower bounds. 
For hexagonal-cells graphs, Chan~{\etal}~\cite{ChChYZ07,ChChYZ10} 
give an incremental algorithm with competitive ratio $1.9216$ and
prove that no ratio better than $1.5$ is possible. A lower bound
of $4/3$ for paths was given in \cite{ChCYZZ06}, and later
Chrobak and Sgall \cite{ChrSga10} gave an incremental algorithm with
the same ratio. Paths are in fact the only non-trivial graphs for which tight
asymptotic ratios are known. As pointed out earlier, there is a strong
connection between frequency allocation and graph coloring, so one would
expect that the competitive ratio can be bounded in terms of the chromatic
number. Indeed, for $\xi$-colorable graphs Chan~{\etal}~\cite{ChChYZ07,ChChYZ10}
give an incremental algorithm with competitive ratio of $(\xi+1)/2$.
(This ratio is in fact absolute.)
On the other hand, the best known lower bounds on the competitive ratio,
$1.5$ in the asymptotic and $2$ in the absolute case~\cite{ChChYZ07,ChChYZ10},
hold for hexagonal-cell graphs, but no stronger bounds are known for graphs
of higher chromatic number.

%%%%%%%%%%%%%%%%%%%%%

\paragraph{Our contribution.}
In this paper, we prove nearly tight bounds on the optimal competitive ratio of
incremental algorithms for bipartite graphs, showing that it is between 
$10/7\approx 1.428$ and $(18-\sqrt{5})/11\approx 1.433$.
This improves the lower and upper bounds for this version of frequency allocation.
The best previously known lower bound was $4/3$, which
holds in fact even for paths \cite{ChCYZZ06,ChrSga10}. The best upper
bound of $1.5$ was shown in \cite{ChChYZ07,ChChYZ10} and it holds even
in the absolute case. 

Our proofs are based on reducing the incremental problem to a purely 
combinatorial (equivalent) problem of 
constructing set families, which we call F-systems, with certain intersection properties.
A rather surprising consequence of this reduction is
that the optimal competitive ratio can be achieved
by an algorithm that is topology-independent; it assigns
a frequency to each vertex $v$ based only on the current optimum value,
the number of requests to $v$, and the partition of the vertex $v$;
that is, independently of the actual frequencies already
assigned to the neighbors of $v$.

To achieve a competitive ratio below $2$ for bipartite graphs, we need to use
frequencies that are shared between the two partitions of the graph. The
challenge is then to assign these shared frequencies to the requests in
different partitions so as to avoid collisions -- in essence, to break
the symmetry. In our construction, we develop a symmetry-breaking method
based on the concepts of ``collisions with the past'' and ``collisions with the
future'', which allows us to derive frequency sets in a systematic fashion.
We believe that these two ideas -- the concept of
F-systems and our symmetry-breaking method -- can be extended to frequency
assignment problems in other types of graphs.

%%%%%%%%%%%%%%%%%%%%%

\paragraph{Other related work.}
Determining optimal absolute ratios
is usually easier than for asymptotic ratios and it has been accomplished for
various classes of graphs, including paths~\cite{ChCYZZ06} and bipartite graphs
in general~\cite{ChChYZ07,ChChYZ10}, and hexagonal-cell graphs and $3$-colorable
graphs in general~\cite{ChChYZ07,ChChYZ10}.
The asymptotic ratio model, however, is more relevant to practical scenarios where the
number of frequencies is typically very large, so the additive constant
can be neglected.

In the \emph{dynamic version} of frequency allocation each request
has an arrival and departure time. At each time, any two requests
that have already arrived but not departed and are in the same
or adjacent nodes must be assigned different frequencies. As before,
we wish to minimize the total number of used frequencies.
As shown by Chrobak and Sgall \cite{ChrSga10}, this dynamic
version is {\NP}-hard even for the special case when the input graph is a path. 

It is natural to study the online version of this problem, where
we introduce the notion of ``time" that progresses in discrete steps, and
at each time step some requests may arrive and some previously arrived
requests may depart. This corresponds to real-life wireless networks where
customers enter and leave a network's cells over time, in an unpredictable
fashion. An online algorithm needs to assign frequencies to requests as 
soon as they arrive. The competitive ratio is defined 
analogously to the incremental case. (The incremental static
version can be thought of as
a special case in which all departure times are infinite.)
This model has been well studied in the context of job scheduling, where it is sometimes
referred to as time-online. Very little is known about this online dynamic
case. Even for paths the optimal ratio is not known; it is only known
that it is between $\fourteennineth\approx 1.571$ \cite{ChrSga10}
and $\fivethirds\approx 1.667$ \cite{ChCYZZ06}.

%%%%%%%%%%%%%%%%%%%%%%%%%%%%%%%%%%%%%%%%%%%%%%%%%%%%%%%%%%%%%%%%%%%%%%%%%%%%%%
%%%%%%%%%%%%%%%%%%%%%%%%%%%%%%%%%%%%%%%%%%%%%%%%%%%%%%%%%%%%%%%%%%%%%%%%%%%%%%
%%%%%%%%%%%%%%%%%%%%%%%%%%%%%%%%%%%%%%%%%%%%%%%%%%%%%%%%%%%%%%%%%%%%%%%%%%%%%%

\section{Preliminaries}

For concreteness, we will assume that frequencies are identified by
positive integers, although it does not really matter.
Recall that we use the number of frequencies as the performance
measure. In some literature~\cite{ChCYZZ06,ChZhZh07,ChChYZ10},
authors used the maximum-numbered frequency instead. It is not hard to show
(see \cite{ChrSga10}, for example, which does however involve a transformation
of the algorithm that makes it \emph{not} topology independent)
that these two approaches are equivalent.    

For a bipartite graph $G = (A,B,E)$, it is easy to characterize the
optimum value. As observed in~\cite{ChCYZZ06,ChrSga10}, in this case
the optimum number of frequencies is
\begin{eqnarray}
		\opt(G,\barell) &=& \max_{(u,v)\in E} \braced{ \ell_u + \ell_v }.
			\label{eqn: static optimum}
\end{eqnarray}
For completeness, we include a simple proof: Trivially,
$\opt(G,\barell) \ge \ell_u+\ell_v$ for each edge $(u,v)$. 
On the other hand, denoting by $\omega$ the right-hand side of
(\ref{eqn: static optimum}), we can assign frequencies to
nodes as follows: for $u\in A$, assign to $u$ frequencies
$1,2,\ldots,\ell_u$, and for $u\in B$ assign to $u$ frequencies
$\omega-\ell_u+1,\omega-\ell_u+2,\ldots,\omega$. This way each vertex $u$
is assigned $\ell_u$ frequencies and no two adjacent nodes
share the same frequency.

Throughout the paper, we will use the convention that
if $c\in\braced{A,B}$, then 
$c'$ denotes the partition other than $c$, that is 
$\braced{c,c'} = \braced{A,B}$.

%%%%%%%%%%%%%%%%%%%%%%%%%%%%%%%%%%%%%%%%%%%%%%%%%%%%%%%%%%%%%%%%%%%%%%%%%%%%%%
%%%%%%%%%%%%%%%%%%%%%%%%%%%%%%%%%%%%%%%%%%%%%%%%%%%%%%%%%%%%%%%%%%%%%%%%%%%%%%
%%%%%%%%%%%%%%%%%%%%%%%%%%%%%%%%%%%%%%%%%%%%%%%%%%%%%%%%%%%%%%%%%%%%%%%%%%%%%%

\section{Competitive F-Systems}

In this section we show that finding an $R$-competitive algorithm
for bipartite graphs can be reduced to an equivalent problem of
constructing certain families of sets that we call F-systems.

Suppose that for any $c\in\braced{A,B}$ and any integers
$t,k$ such that $0 < k \le t$, we are given a set
$F^c_{t,k}$ of positive integers (frequencies).
Denote by $\calF = \bigbrace{F^c_{t,k}}$ the family of those sets.
Then $\calF$ is called an \emph{$F$-system} if
\begin{description}
	\item{(F1)} $|F^c_{t,k}| \ge k$ for all $c,t,k$, and
	\item{(F2)} $F^A_{t,k} \cap F^B_{t',k'} = \emptyset$ 
	 	for all	$k,k',t,t'$ such that $k+k'\le \max(t,t')$.
\end{description}
An F-system is called \emph{$R$-competitive} if
for all $t$ we have
\begin{eqnarray}
	\Bigbars{ \bigcup_{c=A,B}\bigcup_{\kappa \le \tau\le t} F^c_{\tau,\kappa} }
			&\le& R\cdot t + \lambda,
			\label{eqn: def rho-competitive}
\end{eqnarray}
where $\lambda$ is a constant independent of $t$.
The \emph{competititive ratio} of $\calF$ is the smallest $R$ for
which $\calF$ is $R$-competitive.

%%%%%%%%%%%%%%%%%%%%%%

\begin{lemma}\label{lem: incr alg equiv f-system}
For any $R\ge 1$, there is an $R$-competitive incremental algorithm
for frequency allocation in bipartite graphs if and only if there
is an $R$-competitive F-system.
\end{lemma}

\begin{proof}
$(\Rightarrow)$
Let $\calA$ be an $R$-competitive incremental algorithm.
To prove this implication, we define a ``universal" infinite bipartite
graph $H = (A,B,E)$ and we will issue requests to this graph.
For $c \in\braced{A,B}$, the vertices in $c$ have the form $(t,k)_c$,
where $k \le t$. Two vertices $(t,k)_A$ and $(t',k')_B$ are
connected by an edge if $k+k' \le \max(t,t')$. 

The requests are issued in phases numbered $t=1,2,\ldots$. In phase $t$,
for each node $(t,k)_c$, we issue $k$ requests to this node.
Let $F^c_{t,k}$ be the set of frequencies that $\calA$
assigns to $(t,k)_c$. After phase $t$, by the definition of
$H$ and by (\ref{eqn: static optimum}),
the optimum number of frequencies is $t$, so $\calA$
uses at most $Rt+\lambda$ frequencies, for some $\lambda$. 
In other words, (\ref{eqn: def rho-competitive}) holds.
Thus $\calF = \bigbrace{F^c_{t,k}}$ is an $R$-competitive F-system.

$(\Leftarrow)$
Let $\calF$ be an $R$-competitive F-system. We use $\calF$ to define
an incremental algorithm $\calA$ that works as follows. 
Let $G = (A,B,E)$ be the given bipartite graph. Consider one step
of the computation in which a new request arrives at a
vertex $u \in c$, where $c\in\braced{A,B}$. Denote by $t$
the current optimum number of frequencies, that is
$t = \max_{(v,w)\in E}(\ell_v + \ell_w)$. Choose any frequency
$f\in F^c_{t,k}$, for $k = \ell_u$, that is not yet used on $u$
and assign $f$ to this request. Such $f$ exists, because by
property (F1) we have
$|F^c_{t,k}|\ge k$ and the number of frequencies assigned so far 
to $u$ is $k-1$. 

Trivially, all frequencies assigned by $\calA$ to one node are different.
We claim that adjacent nodes will be assigned
different frequencies as well. Consider again a step where a frequency
$f$ is assigned to a $k$th request to a vertex $u$, 
when the optimum value is $t$, as described above. So $k = \ell_u$.
Without loss of generality, assume $u\in A$. For an edge $(u,v)\in E$,
let $k' = \ell_v$ be the current load at $v$. If
$g$ is any frequency assigned by $\calA$ to $v$ then, by
the definition of $\calA$, we have that
$g\in F^B_{t',k''}$ for some $t'\le t$ and $k''\le \min(t', k')$.
Thus $k+k'' \le k+k' \le t$, by the definition of $t$.
Using condition (F2), we now get that
$F^A_{t,k}\cap F^B_{t',k''} = \emptyset$, and therefore $f\neq g$.

Finally, when the optimum value is $t$, then any frequency
used by $\calA$ is from some set $F^c_{\tau,\kappa}$ for
$\kappa\le\tau\le t$. Therefore $\calA$ is $R$-competitive, by the
property (\ref{eqn: def rho-competitive}) of $\calF$.
\end{proof}

%%%%%%%%%%%%%%%%%%%%%%%%%%%%%%%%%%%%%%%%%%%%%%%%%%%%%%%%%%%%%%%%%%%%%%
%%%%%%%%%%%%%%%%%%%%%%%%%%%%%%%%%%%%%%%%%%%%%%%%%%%%%%%%%%%%%%%%%%%%%%
%%%%%%%%%%%%%%%%%%%%%%%%%%%%%%%%%%%%%%%%%%%%%%%%%%%%%%%%%%%%%%%%%%%%%%

\section{An Upper Bound}
\label{sec: an upper bound}

In this section we prove that there is an $R_0$-competitive
incremental algorithm, for $R_0 = (18-\sqrt{5})/11\approx 1.433$.
Using Lemma~\ref{lem: incr alg equiv f-system}, it is sufficient to
design an $R_0$-competitive F-system.

%%%%%%%%%%%%%%%%%

\paragraph{Intuitions.}
Our construction below may appear rather mysterious, so we begin by
gradually introducing its main ideas.
We will distinguish between two types of frequencies: private and
shared. A-private frequencies will be used only in sets $F^A_{t,k}$,
B-private frequencies will be used only in sets $F^B_{t,k}$,
while shared frequencies can be used in some sets from both partitions
$A$ and $B$.

Competitive ratio $2$ can be easily achieved using only private frequencies. For
each $c\in \braced{A,B}$,
let $P^c$ denote an infinite pool of $c$-private frequencies, 
with $P^A$ and $P^B$ disjoint. We
simply let $F^c_{t,k}$ consist of the first $k$ frequencies from $P^c$.
Conditions (F1) and (F2) are trivially true.
For any given $t$, the set on the left-hand side of 
inequality (\ref{eqn: def rho-competitive}) contains $2t$ frequencies,
so (\ref{eqn: def rho-competitive}) holds for $R = 2$.

We now show how to improve the ratio to $1.5$. To accomplish this,
we must use some shared frequencies.
Let $S$ denote an infinite pool of shared frequencies, where $S$ is
disjoint with $P^A\cup P^B$.
To avoid collisions (that is, violations of (F2)), we need
to use these shared frequencies judiciously. The main idea is this: for
any given $c,t,k$, $F^c_{t,k}$ will only contain some of the first
$t/2$ $c$-private frequencies and some of the first $t/2$ shared
frequencies. (For simplicity, we temporarily
ignore the fact that $t/2$ may not be integer.) 
This will guarantee that we will use at most $1.5t$
frequencies for all sets $F^c_{\tau,\kappa}$ with $\tau\le t$.
If $k\le t/2$, then we have enough $c$-private frequencies to completely fill
$F^c_{t,k}$. Otherwise, for $k > t/2$,
in addition to the first $t/2$ $c$-private frequencies,
$F^c_{t,k}$ we use $k-t/2$ \emph{last} shared
frequencies with indices at most $t/2$. So these frequencies will be
indexed between $t/2 - (k-t/2) = t-k$ and $t/2$. Clearly,
$F^c_{t,k}$ has at least $k$ frequencies, so (F1) holds.
The intuition behind (F2) is this: Suppose $t'\le t$. Then
$F^A_{t,k}$ conflicts with each $F^B_{t',k'}$ for $k' \le t-k$.
As $k'\le t-k$, the ``worst" such conflict is with $F^B_{t-k,t-k}$, which 
is disjoint with $F^A_{t,k}$, by our choice of shared frequencies.

To make it more precise, for any real number $x\ge 0$ let
\begin{eqnarray*}
	S_x &=& \textrm{the first $\floor{x}$ frequencies in $S$ } , \\
	P^c_x &=& \textrm{the first $\floor{x}$ frequencies in $P^c$, for $c\in\braced{A,B}$}.
\end{eqnarray*}
We now let $\calF = \bigbrace{F^c_{t,k}}$, where for $c\in\braced{A,B}$ and $k\le t$ we
have
\begin{eqnarray*}
	F^c_{t,k} &=& P^c_{t/2+1} \cup (S_{t/2} \setminus S_{t-k} ).
\end{eqnarray*}
We claim that $\calF$ is a $1.5$-competitive F-system.
If $k\le \floor{t/2}+1$, then $|F^c_{t,k}|\ge k$ is trivial.
If $k\ge \floor{t/2}+2$, then $t- k \le t - \floor{t/2}- 2 \le  t/2$,
so $S_{t-k}\subseteq S_{t/2}$ and thus
$|F^c_{t,k}| \ge \floor{t/2} + 1 + ( \floor{t/2} - \floor{t - k} )
	\ge k$. So (F1) holds. 

To verify (F2),
pick any two pairs $k\le t$ and $k'\le t'$ with $k+k'\le\max(t,t')$. 
Without loss of generality, assume $t'\le t$ and $c = A$. 
If $k'\le \floor{t'/2}+1$, then $F^B_{t',k'} \subseteq P^B$, so
(F2) is trivial. If $k'\ge \floor{t'/2}+2$, then $t'/2\le k' \le t-k$, so
$F^B_{t',k'} \subseteq P^B \cup S_{t'/2}
			\subseteq P^B \cup S_{t-k}$,
which implies (F2) as well.

Finally, for any $c\in\braced{A,B}$ and $\kappa\le\tau\le t$, we have
$F^c_{\tau,\kappa} \subseteq P^A_{t/2+1}\cup P^B_{t/2+1} \cup S_{t/2}$,
so the inequality (\ref{eqn: def rho-competitive}) holds with 
$R = 1.5$ and $\lambda = 2$. We can thus conclude that this
$\calF$ is $1.5$-competitive.

\medskip

A geometric interpretation of the used sets of frequencies is is shown
in Figure~\ref{fig: shared frequencies 1.5}.  For $k > t/2$, set
$F^c_{t,k}$ conflicts with $F^{c'}_{\tau,\tau}$ for $\tau = t-k$ and
$F^{c'}_{\tau,\tau}$ uses shared frequencies numbered at most $\tau'/2
= (t-k)/2$. Thus all shared frequencies that ``conflict with the past"
are within the region below the line $x=(t-k)/2$. This region is
disjoint with the shaded region assigned to $F^{c'}_{t,k}$, whose
boundary is the line $x=t-k$.

%%%%%%%%%%%%%%%%

\begin{figure}[t]
\begin{center}
\includegraphics[width=4in]{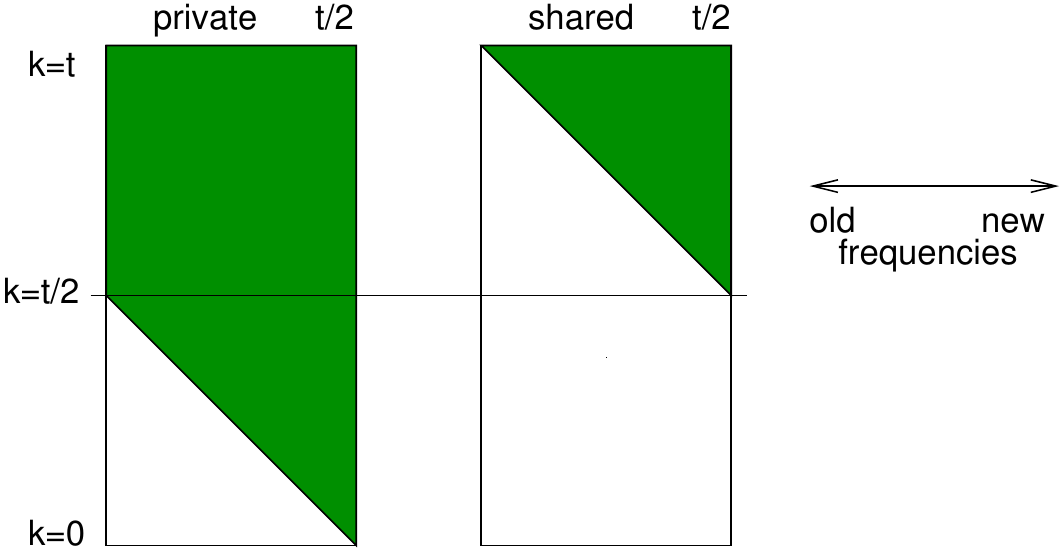}

\caption{The structure of frequency sets in the $1.5$-competitive
algorithm.  In this figure, we fix the value of $t$, and show the
frequency sets for each value of $k\leq t$.  The horizontal axis
represents frequencies, with the first frequencies drawn on the left.
The vertical axis represents the values of $k$, for each $k$ the
intersection of the corresponding horizontal line with the shaded
(green) regions shows the frequencies used by the algorithm.  For
private frequencies on the left, for $k<t/2$, we do not need to use
all of the frequencies, and the choice of them is arbitrary.  For
shared frequencies on the right, the shaded area corresponds exactly
to $S_{t/2} \setminus S_{t-k}$.  }
\label{fig: shared frequencies 1.5}
\end{center}
\end{figure}

%%%%%%%%%%%%%%%%%

\paragraph{Construction of an $R_0$-competitive F-system.}

To improve the ratio further, the idea is to use even fewer private
frequencies, but to assign shared frequencies more carefully. We will
actually have three types of shared frequencies, called A-shared,
B-shared and symmetric-shared.  

To achieve ratio smaller than $1.5$ we need to use some shared
frequencies even for $k<t/2$. Obviously, to do this we must break
symmetry, as $F^A_{t,k}$ and $F^B_{t,k}$ cannot use any common shared
frequency for $k<t/2$. This is the reason why we introduce A-shared
and B-shared frequencies.  For sets $F^c_{t,k}$, as $k$ increases, we
first use $c$-private frequencies, then $c$-shared frequencies, then
symmetric-shared frequencies, and finally, if $k$ gets sufficiently
large, we also ``borrow" $c'$-shared frequencies to include in this
set. More precisely, we use some $c$-shared frequencies for any
$k>t/\phi^2\approx 0.382 t$, while we use symmetric-shared frequencies
for $k > t/2$ and $c'$-shared frequencies only for $k>t/\phi\approx
0.618 t$.  We remark here that symmetric frequencies are still
needed. If we restrict ourselves to only private and $c$-shared
frequencies then the best ratio we are able to achieve is $\approx
1.447$.

\begin{figure}[t]
\begin{center}
\includegraphics[width=4in]{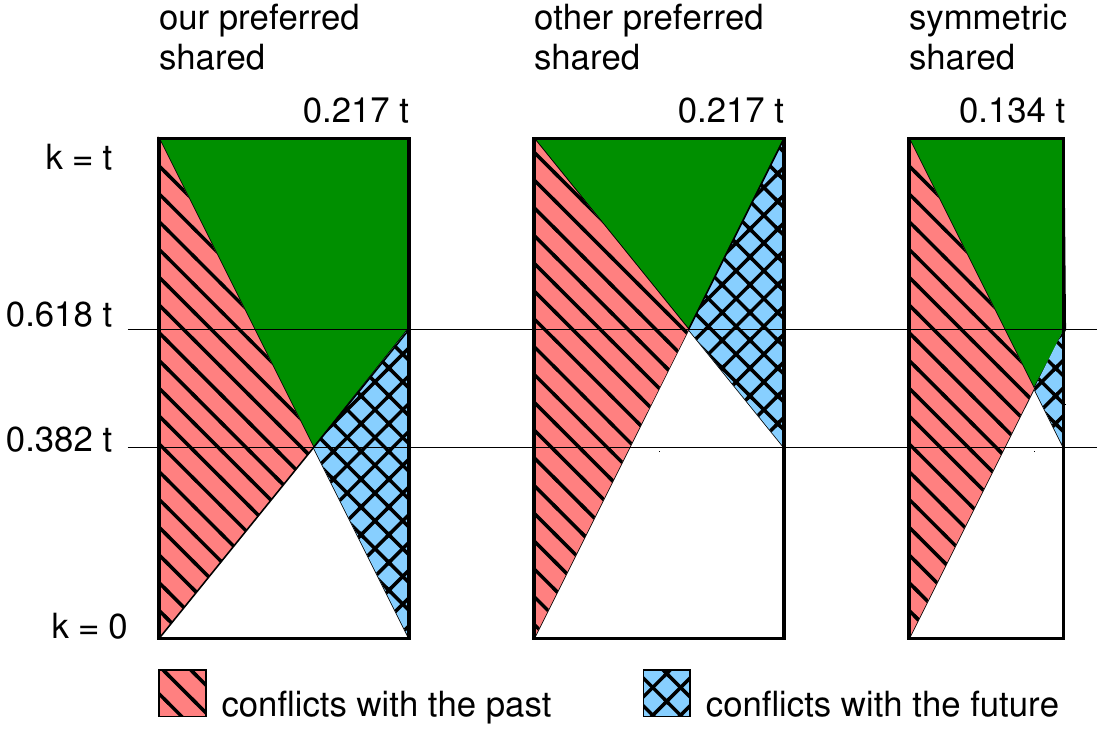}

\caption{The structure of frequency sets in the $R_0$-competitive
algorithm. Here we show only the shared frequencies, represented
similarly as in Figure~\ref{fig: shared frequencies 1.5}. In addition,
we show by different shading which of the unused frequencies would
create conflits with the past and with the future; the bottom unshaded
part would cause both types of conflicts.}
\label{fig: shared frequencies 1.43}
\end{center}
\end{figure}

A geometric interpretation of the used sets of shared frequencies is
is shown in Figure~\ref{fig: shared frequencies 1.43}.  The algorithm
with ratio $1.5$ used the shaded region shown in Figure~\ref{fig:
shared frequencies 1.5}, to avoid collisions with the past, that is
with frequencies already assigned to sets $F^c_{\tau,\tau}$ for $\tau
< t$. As observed earlier, the line $x=t-k$ is not tight; it can be
lowered to $x=(t-k)/2$ without creating conflicts. With this
modification, only some of the shared frequencies above this line are
needed. However, this modification is not sufficient to reduce the
ratio below $1.5$, because of symmetry: we still will have conflicts
for $F^{c'}_{t,k}$ for $k = t/2$.  To avoid such conflicts, we also
consider, preemptively, ``conflicts in the future", namely with sets
$F^c_{t',k'}$ for $t' > k$. These conflicts are represented in the
figure by the half-plane below the line $x=\gamma k$, for an
appropriate $\gamma$ while now the conflicts with the past are
represented by the half-plane below the line $\gamma'(t-k)$. The
optimization of the parameters for all three types of shared
frequencies leads to our new algorithm.

\medskip

The pools of $c$-shared and symmetric-shared frequencies are denoted
$S^c$ and $Q$, respectively. As before, for any real $x\ge 0$ we define
\begin{eqnarray*}
	S^c_x &=& \textrm{the first $\floor{x}$ frequencies in $S^c$, for $c\in\braced{A,B}$}.
\\
	Q_x &=& \textrm{the first $\floor{x}$ frequencies in $Q$}.
\end{eqnarray*}
Our construction uses three constants, defined as
\begin{eqnarray*}
\alpha&=&R_0-1=\frac{7-\sqrt{5}}{11}=\frac{2}{\phi+3}\approx 0.433,
\\
\beta&=&\alpha/2=\frac{7-\sqrt{5}}{22}=\frac{1}{\phi+3}\approx 0.217,
\mbox{ and}
\\
\rho&=&\beta/\phi=\frac{2\sqrt{5}-3}{11}=\frac{\phi-1}{\phi+3}\approx 0.134,
\end{eqnarray*}
where $\phi=(\sqrt5+1)/2$ is the golden ratio.
A useful fact is the identity $2\alpha+2\beta+\rho=R_0$.

We define $\calF = \bigbrace{ F^c_{t,k} }$, where for 
any $t\geq k\geq 0$ we let
\begin{eqnarray}
F^c_{t,k} &=& P^c_{\alpha t+4} 
			\cup  (S^c_{\beta\cdot\min(t,\phi k)} \setminus S^c_{\beta(t-k)})
			\cup (S^{c'}_{\beta k} \setminus S^{c'}_{\phi\beta(t-k)})
			\cup  (Q_{\rho\cdot\min(t,\phi k)} \setminus Q_{\phi\rho(t-k)}).
			\label{eqn: our F-system}
\end{eqnarray}

We now show that $\calF$ is an $R_0$-competitive F-system. To this end, 
we show that $\calF$ satisfies properties (F1), (F2) and
(\ref{eqn: def rho-competitive}).

We start with (\ref{eqn: def rho-competitive}).
For $\kappa\le \tau \le t$ and $c\in\braced{A,B}$ we have
\begin{eqnarray*}
F^c_{\tau,\kappa} &\subseteq& P^c_{\alpha\tau+4} \cup S^c_{\beta \tau}
\cup S^{c'}_{\beta \kappa} \cup Q_{\rho\tau}
		\;\subseteq\; P^c_{\alpha t+4} \cup S^c_{\beta t} \cup
                S^{c'}_{\beta t}\cup Q_{\rho t}
		\;\subseteq\; P^A_{\alpha t+4} \cup P^B_{\alpha t+4}
						\cup S^A_{\beta t} \cup S^B_{\beta t}\cup Q_{\rho t}.
\end{eqnarray*}
This last set has cardinality at most
$(2\alpha+2\beta+\rho)t+8 = R_0t+8$,  
so (\ref{eqn: def rho-competitive}) holds with $\lambda = 8$.

\smallskip

Next, we show (F2). By symmetry, we can assume that $t'\le t$ in (F2),
so $k' \le t-k$. Then
\[
F^{B}_{t',k'} 
	\;\subseteq\; P^{B} \cup S^{B}_{\phi\beta k'}\cup S^A_{\beta k'}\cup Q_{\phi\rho k'}
	\;\subseteq\; P^{B} \cup S^{B}_{\phi\beta (t-k)}\cup S^A_{\beta(t-k)}\cup Q_{\phi\rho(t-k)},
\]
and this set is disjoint with $F^A_{t,k}$ by definition~(\ref{eqn: our F-system}). 
Thus $F^A_{t,k}\cap F^{B}_{t',k'} = \emptyset$, as needed.

\smallskip

Finally, we prove (F1), namely that $|F^c_{t,k}|\ge k$. We
distinguish two cases.
\begin{description}

\item{{\mycase{1}}} $k>t/\phi$. 
This implies that $\min(t, \phi k) = t$, so in (\ref{eqn: our F-system}) we have
$S^c_{\beta\cdot\min(t,\phi k)}=S^c_{\beta t}$ and $Q_{\rho\cdot\min(t,\phi k)}=Q_{\rho t}$.
Thus
\begin{eqnarray*}
|F^c_{t,k}| &\geq &	[ \alpha t+3 ] + [ \beta t- \beta(t-k)-1 ] + 
	[ \beta k-\phi\beta(t-k)-1 ] + [\rho t -\phi\rho (t-k)-1 ] 
			\\ 
&=& (\alpha-\phi\beta-(\phi-1)\rho)t+ (2\beta+\phi\beta+\phi\rho)k
\;=\; k,
\end{eqnarray*}
using the substitutions $\alpha=2\beta$ and $\rho=\beta/\phi$.
Note that this case is asymptotically
tight as the algorithm uses all
three types of shared frequencies (and the corresponding terms are
non-negative). 

\item{{\mycase{2}}} $k\leq t/\phi$. 
The case condition implies that
$\phi k \le t$, so $S^c_{\beta\cdot\min(t,\phi k)}=S^c_{\phi\beta k}$,
$Q_{\rho\cdot\min(t,\phi k)}=Q_{\phi\rho k}$, and 
$S^{c'}_{\beta k} \setminus S^{c'}_{\phi\beta(t-k)} = \emptyset$. Therefore
\begin{eqnarray*}
|F^c_{t,k}| &\geq& [ \alpha t+ 3 ] + [ \phi\beta k- \beta (t-k)-1 ]
           + [\phi\rho k -\phi\rho (t-k)-1 ] 
\\
		&=&  (\alpha - \beta-\phi\rho) t + ((\phi+1)\beta+2\phi\rho) k+1
		\\
		&=& k+1.
\end{eqnarray*}
using $\alpha=2\beta$ and $\rho=\beta/\phi$ again. Note that this case is
(asymptotically) tight only for $k>t/2$ when $c$-shared and
symmetric-shared frequencies are used. For $k\leq t/2$, no
symmetric-shared frequencies are used and the corresponding term is
negative. 
\end{description}

Summarizing, we conclude that $\calF$ is indeed an $R_0$-competitive
F-system. Therefore, using Lemma~\ref{lem: incr alg equiv f-system}, we get
our upper bound:

%%%%%%%%%%%%%%%

\begin{theorem}
There is an $R_0$-competitive incremental
algorithm for frequency allocation on bipartite graphs, where
$R_0=(18-\sqrt{5})/11 \approx 1.433$.
\end{theorem}

%%%%%%%%%%%%%%%%%%%%%%%%%%%%%%%%%%%%%%%%%%%%%%%%%%%%%%%%%%%%%%%%%%%%
%%%%%%%%%%%%%%%%%%%%%%%%%%%%%%%%%%%%%%%%%%%%%%%%%%%%%%%%%%%%%%%%%%%%
%%%%%%%%%%%%%%%%%%%%%%%%%%%%%%%%%%%%%%%%%%%%%%%%%%%%%%%%%%%%%%%%%%%%

\section{A Lower Bound}
\label{sec: a lower bound}

In this section we show that if $R < 10/7$, then there is no $R$-competitive 
incremental algorithm for frequency allocation in bipartite graphs.
By Lemma~\ref{lem: incr alg equiv f-system}, it is sufficient to
show that there is no $R$-competitive F-system. 

The general intuition behind the proof is that we try to reason about
the sets $Z_t=F^A_{t,\gamma t}\cup F^B_{t,\gamma t}$ for a suitable
constant $\gamma$.  These sets should correspond to the
symmetric-shared frequencies from our algorithm, for $\gamma$ such
that no $c'$-shared frequencies are used. If $Z_t$ is too small, then
both partitions use mostly different frequencies and this yields a
lower bound on the competitive ratio. If $Z_t$ is too large, then for
a larger $t$ and suitable case, the frequencies cannot be used for
either partition, and hopefully this allows to improve the lower
bound. We are not able to do exactly this. Instead, for a variant of
$Z_t$, we show a recurrence essentially saying that if the set is too
large, then for some larger $t$, it must be proportionally even
larger, leading to a contradiction.

We now proceed with the proof.  For $c\in\braced{A,B}$, let $F^c_t =
\bigcup_{\kappa\le\tau\le t} F^c_{\tau,\kappa}$. Towards
contradiction, suppose that an F-system $\calF$ is $R$-competitive for
some $R < 10/7$. Then $\calF$ satisfies the definition of
competitiveness (\ref{eqn: def rho-competitive}) for some positive
integer $\lambda$. Choose a sufficiently large integer $\theta$ for
which $R < 10/7 - 1/\theta$.

We first identify shared frequencies in $\calF$. Recall that 
$F^c_t = \bigcup_{\kappa\le\tau\le t} F^c_{\tau,\kappa}$,  for $c\in\braced{A,B}$.
Thus the definition of $R$-competitiveness says that
$|F^A_t\cup F^B_t| \le Rt+\lambda$. The set of \emph{level-$t$ shared} 
frequencies is defined as $S_t = F^A_t\cap F^B_t$.

%%%%%%%%%%%

\begin{lemma}\label{lem: shared lower bound}
For any $t$, we have $|S_t|\geq (2-R)t-\lambda$.
\end{lemma}

\begin{proof}
This is quite straightforward.  By (F1) we have
$|F^c_t|\ge t$ for each $c$, so
$|S_t| = |F^A_t|+|F^B_t| - |F^A_t\cup F^B_t|
	\geq 2t-(Rt+\lambda) 
	= (2-R)t-\lambda$.
\end{proof}

Now, let 
$S_{2t,t} = S_{2t}\cap (F^A_{2t,t}\cup F^B_{2t,t})$ be the level-$2t$ shared frequencies
that are used in $F^A_{2t,t}$ or $F^B_{2t,t}$.
Each such frequency can only be in one of these sets because 
$F^A_{2t,t} \cap F^B_{2t,t} = \emptyset$.

%%%%%%%%%%%

\begin{lemma}\label{lem: S_{2t,t} lower bound}
For any $t$, we have $|S_{2t,t}| \ge (6-4R)t - 2\lambda$.
\end{lemma}

\begin{proof}
Observe that $F^A_{2t,t} \cup F^B_{2t,t} \cup S_{2t} \subseteq F^A_{2t}\cup F^B_{2t}$
by definition, and thus~\eqref{eqn: def rho-competitive} implies
\begin{align*}
2Rt+\lambda 	&\geq |F^A_{2t,t}\cup F^B_{2t,t}\cup S_{2t}|\\
		&= |F^A_{2t,t}\cup F^B_{2t,t}| + |S_{2t}| - |(F^A_{2t,t}\cup F^B_{2t,t}) \cap S_{2t}|\\
		&= |F^A_{2t,t}| + |F^B_{2t,t}| + |S_{2t}| - |S_{2t,t}| \enspace,
\end{align*}
where the identities follow from the inclusion-exclusion principle, disjointness of
$F^A_{2t,t}$ and $F^B_{2t,t}$, and the definition of $S_{2t,t}$. 

Transforming this inequality, we get 
\begin{align*}
|S_{2t,t}| 	&\geq |F^A_{2t,t}| + |F^B_{2t,t}| + |S_{2t}| -(2Rt+\lambda)\\
		&\geq (6-4R)t-2\lambda \enspace,
\end{align*}
as claimed, by property~(F1) and Lemma~\ref{lem: shared lower bound}.
\end{proof}

For any even $t$ define $Z_{3t/2,t} = F^A_{3t/2,t} \cap F^B_{3t/2,t}$.
In the rest of the lower-bound proof we will set up a recurrence
relation for
the cardinality of sets $S_{t}\cup Z_{3t/2,t}$. The next step
is the following lemma.

%%%%%%%%%%%%

\begin{lemma}\label{lem: subsets of S_2t - Z_2t}
For any even $t$, we have
$|S_{2t} \setminus Z_{3t,2t}| \ge |S_t\cup Z_{3t/2,t}| + |S_{2t,t}|$.
\end{lemma}

\begin{figure}[t]
\begin{center}
\includegraphics[width=3.5in]{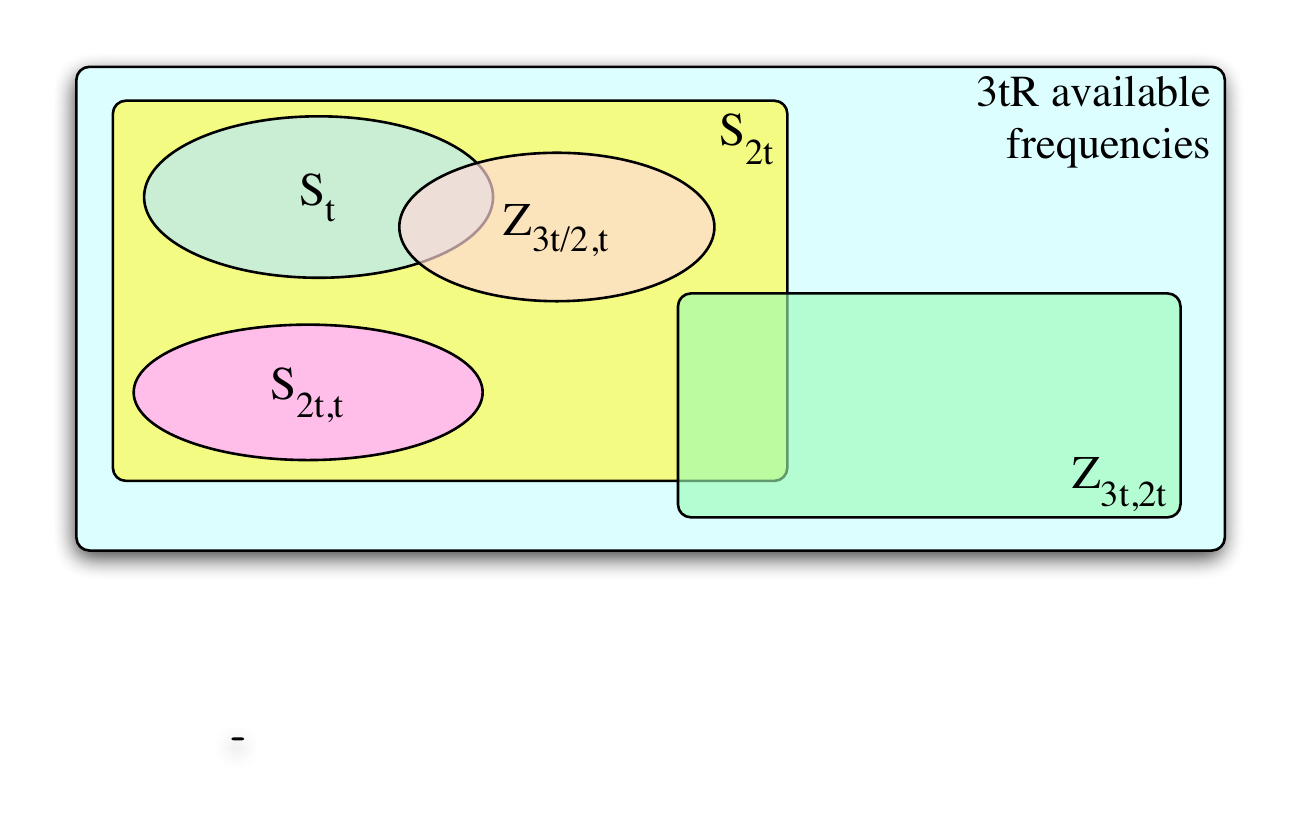}
%\vspace{-0.5in}
\caption{Illustration of Lemma~\ref{lem: subsets of S_2t - Z_2t}.}
\label{fig: subsets of S_2t - Z_2t}
\end{center}
\end{figure}

\begin{proof}
From the definition, the two sets
$S_{t}\cup Z_{3t/2,t}$ and $S_{2t,t}$ are disjoint
and they are subsets of $S_{2t} - Z_{3t,2t}$.
(See Figure~\ref{fig: subsets of S_2t - Z_2t} for illustration.)
This immediately implies the lemma. 
\end{proof}

%%%%%%%%%%%%

\begin{lemma}\label{lem: 3Rt + ... >= 4t}
For any even $t$, we have
$|Z_{3t,2t}| \ge |S_{t} \cup Z_{3t/2,t}| - (3R-4)t - \lambda$.
\end{lemma}

\begin{proof}
As $F^A_{3t,2t}$, $F^B_{3t,2t}$, $S_t$ and $Z_{3t/2,t}$ are all subsets of $F^A_{3t}\cup F^B_{3t}$,
inequality~\eqref{eqn: def rho-competitive} implies
\begin{align*}
3Rt+\lambda 	&\geq |F^A_{3t,2t}\cup F^B_{3t,2t}\cup S_t\cup Z_{3t/2,t}|\\
		&= |F^A_{3t,2t}\cup F^B_{3t,2t}| + |S_t \cup Z_{3t/2,t}|\\
		&= |F^A_{3t,2t}| + |F^B_{3t,2t}| - |F^A_{3t,2t}\cap F^B_{3t,2t}|+ |S_t \cup Z_{3t/2,t}|\\
		&= |F^A_{3t,2t}| + |F^B_{3t,2t}| - |Z_{3t,2t}|+ |S_t \cup Z_{3t/2,t}| \enspace,
\end{align*}
where the identities follow from the inclusion-exclusion principle, the fact that
$F^A_{3t,2t}\cup F^B_{3t,2t}$ and $S_t\cup Z_{3t/2,t}$ are disjoint, and the definition
of $Z_{3t,2t}$.

Transforming this inequality, we get
\begin{align*}
|Z_{3t,2t}| 	&\geq |F^A_{3t,2t}| + |F^B_{3t,2t}| + |S_t\cup Z_{3t/2,t}| -(3Rt+\lambda)\\
		&\geq |S_t\cup Z_{3t/2,t}| - (3R-4)t -\lambda \enspace,
\end{align*}
as claimed, by property~(F1).
\end{proof}

We are now ready to derive our recurrence. By adding the inequalities in
Lemma~\ref{lem: subsets of S_2t - Z_2t} and Lemma~\ref{lem: 3Rt + ... >= 4t},
taking into account that $|S_{2t} \setminus Z_{3t,2t}|$ + $|Z_{3t,2t}|
= |S_{2t} \cup Z_{3t,2t}|$,
and then applying Lemma~\ref{lem: S_{2t,t} lower bound},
for any even $t$ we get
\begin{eqnarray}
|S_{2t} \cup Z_{3t,2t}| 
	&\ge& 2\cdot |S_{t} \cup Z_{3t/2,t}| + |S_{2t,t}| - (3R-4)t - \lambda
		\nonumber
		\\
	&\ge&  2\cdot |S_{t} \cup Z_{3t/2,t}| + (10-7R)t - 3\lambda.
		\label{eqn: lb main recurrence}
\end{eqnarray}

For $i = 0,1,\ldots,\theta$, define $t_i = 6\theta\lambda 2^i$ and
$\gamma_i = |S_{t_i} \cup Z_{3t_i/2,t_i}|/t_i$. (Note that each
$t_i$ is even.)
Since $S_{t_i} \cup Z_{3t_i/2,t_i} \subseteq S_{2t_i}$, we have
that $\gamma_i \le |S_{2t_i}|/t_i \le 2R +1/(6\theta) < 3$.
Dividing recurrence (\ref{eqn: lb main recurrence}) by $t_{i+1} = 2t_i$, 
we obtain, for $i = 0,1,\ldots,\theta-1$,
\begin{eqnarray*}
\gamma_{i+1} &\ge& \gamma_i + 5 - 7R/2 - 3\lambda/(2t_i) 
	\\
	&\ge& \gamma_i + 7/(2\theta) - 1/(4\theta)
	\;\ge\; \gamma_i + 3/\theta.
\end{eqnarray*}
But then we have $\gamma_\theta \ge \gamma_0 + 3 \ge 3$,
which contradicts our earlier bound $\gamma_i < 3$, completing
the proof. Thus we have proved the following.

%%%%%%%%%%%%%%%%%

\begin{theorem}\label{thm: 10/7 lower bound}
If $\calA$ is an $R$-competitive incremental algorithm for frequency
allocation on bipartite graphs, then $R\geq 10/7\approx 1.428$.
\end{theorem}

As a final remark we observe that our lower bound works even if the
additive constant $\lambda$ is allowed to depend on the actual
graph. I.e., for every $R<10/7$ we can construct a single finite graph
$G$ so that no algorithm is $R$-competitive on this graph.  In
 our lower bound, we can restrict our attention to sets
$F^c_{t_i,t_i}$, $F^c_{2t_i,t_i}$, $F^c_{3t_i,2t_i}$, and $F^c_{3t_i/2,t_i}$,
for $i = 0,1,\ldots,\theta$ and $c=A,B$.  Then, in the construction from
the proof of Lemma~\ref{lem: incr alg equiv f-system} for
the lower bound sequence we obtain a finite graph together with a
request sequence. However, for a fixed $\theta$, the
graphs for different values of $\lambda$ are isomorphic, as all the
indices scale linearly with $\lambda$. So, instead of using different
isomorphic graphs, we can use different sequences corresponding to
different values of $\lambda$ on a single graph $G$.

%%%%%%%%%%%%%%%%%%%%%%%%%%%%%%%%%%%%%%%%%%%%%%%%%%%%%%%%%%%%%%%%%%%%%%
%%%%%%%%%%%%%%%%%%%%%%%%%%%%%%%%%%%%%%%%%%%%%%%%%%%%%%%%%%%%%%%%%%%%%%
%%%%%%%%%%%%%%%%%%%%%%%%%%%%%%%%%%%%%%%%%%%%%%%%%%%%%%%%%%%%%%%%%%%%%%

\section{Final Comments}

We proved that the competitive ratio for incremental frequency
allocation on bipartite graphs is between $1.428$ and $1.433$,
improving the previous bounds of $1.33$ and $1.5$. Closing the
remaining gap, small as it is, remains an intriguing open problem. 
Besides completing the analysis of this special case, the
solution is likely to involve sophisticated techniques that may be of 
its own interest.

The two other obvious directions of study are to prove better bounds
for the dynamic case and for $k$-partite graphs. The general idea
of distinguishing ``collisions with the past'' and 
``collisions with the future'', that we use to define our 
frequency sets, should be useful to derive upper bounds for these
problems. Our concept of
F-systems can be extended in a natural way to $k$-partite graphs,
but with a caveat: for $k\ge 3$ the maximum load on a $k$-clique
is only a lower bound on the optimum (unlike for $k=2$, where the
equality holds).
Therefore in Lemma~\ref{lem: incr alg equiv f-system} only one
direction holds. This lemma is still sufficient though to establish
upper bounds on the competitive ratio. It is also conceivable that
a lower bound can be proved using graphs where the optimum
number of frequencies is equal to the maximum load of a $k$-clique.

\bibliographystyle{abbrv}
\bibliography{online,other}

\end{document}

%% file: macros.tex
\newcommand{\calA}{{\cal A}}
\newcommand{\calF}{{\cal F}}

\newcommand{\barell}{{\bar\ell}}

\newcommand{\etal}{{\em et al.}}

\newcommand{\fourthirds}{{\mbox{$\frac{4}{3}$}}}
\newcommand{\fivethirds}{{\mbox{$\frac{5}{3}$}}}

\newcommand{\fourteennineth}{{\mbox{$\frac{14}{9}$}}}

\newcommand{\braced}[1]{{ \left\{ #1 \right\} }}
\newcommand{\bigbrace}[1]{{ \big\{ #1 \big\} }}

\newcommand{\Bigbars}[1]{{ \Big| #1 \Big| }}

\newcommand{\floor}[1]{{ \lfloor #1 \rfloor }}

\newcommand{\opt}{{\textit{opt}}}

\newcommand{\NP}{{\textsf{NP}}}

\newcommand{\mycase}[1]{\mbox{{\underline{Case #1}}:\/}}

\newtheorem{theorem}{Theorem}[section]

\newtheorem{lemma}[theorem]{Lemma}

                {$\Box$\vskip 0.1in}

\newenvironment{bigeqn*}{\large\begin{eqnarray*}}{\end{eqnarray*}}

\newcommand{\ignore}[1]{}